\documentclass[letterpaper,10pt,conference]{ieeeconf}     

\IEEEoverridecommandlockouts                              

\usepackage{amsmath}
\usepackage{amssymb}

\usepackage{algpseudocode}
\usepackage{algorithm}
\usepackage{xcolor}
\usepackage{booktabs}
\usepackage{graphicx} 
\usepackage[normalem]{ulem}

\newtheorem{theorem}{Theorem}

\newtheorem{definition}{Definition}
\newtheorem{problem}{Problem}

\newtheorem{remark}{Remark}

\newcommand{\R}{\mathbb{R}}
\usepackage{physics}
\usepackage{needspace}
\usepackage{comment}

\title{\LARGE \bf
Hierarchical Defense in Leader–Defender–Attacker Games: Nested Equilibria and Receding-Horizon Control
}

\author{Honglei Guo$^{1,3}$, Hanfang Zhang$^{2,3}$, Guochuan Zhang$^{1}$ and Yuhan Zhao$^{3,*}$%
\thanks{$^{1}$Honglei Guo and Guochuan Zhang are with the College of Artificial Intelligence, Zhejiang University, Hangzhou 310027, China. (e-mail: \texttt{guohonglei@bigai.ai};
\texttt{zgc@zju.edu.cn})}%
\thanks{$^{2}$Hanfang Zhang  is with the School of Control Science and Engineering, Shandong University, Jinan 250061, China.(e-mail:\texttt{hanfangzhang@mail.sdu.edu.cn} ).}
\thanks{$^{3}$Honglei Guo, Hanfang Zhang, and Yuhan Zhao are with the State Key Laboratory of General Artificial Intelligence, Beijing Institute for General Artificial  Intelligence (BIGAI), Beijing, China .}%
\thanks{$^{*}$Corresponding authors: Yuhan Zhao (e-mail:
\texttt{yuhanzhao@zju.edu.cn};
).}%
}

\begin{document}

\maketitle
\thispagestyle{empty}
\pagestyle{empty}

\begin{abstract}

We formulate a discrete-time leader–defender–attacker game that combines hierarchical interaction within a protective group with competition against an external attacker. The leader seeks to reach a prescribed demand point while avoiding capture, whereas the defender seeks to intercept the attacker while remaining near the leader. To capture their distinct objectives and asymmetric interactions, we introduce a nested equilibrium concept combining a leader–defender Stackelberg relationship with a Nash-type best-response condition for the attacker. For linear dynamics and quadratic objectives, we derive affine response laws and a backward recursion, together with sufficient conditions for unique stagewise solutions. We further express arrival, capture, and interception conditions as quadratic inequalities in the initial joint state, characterizing winning regions under the resulting policy. A receding-horizon algorithm implements the computed control laws using updated state information. Numerical experiments over randomized initial configurations examine the influence of the prediction horizon and compare the proposed method with an independent-control baseline, illustrating differences in game outcomes and the agents’ realized costs.

\end{abstract}

\section{INTRODUCTION}

Pursuit--evasion games constitute an important class of multi-agent dynamic games. Because the agents move subject to dynamical constraints, optimal control provides a systematic foundation for designing policies that optimize
prescribed performance objectives ~\cite{lewis2012optimal,bertsekas1995dynamic}. When multiple agents pursue distinct objectives, however, their control policies become strategically coupled: each agent's performance depends on both its own decisions and
those of the other agents. Dynamic game theory provides a natural framework for modeling such interactions ~\cite{bacsar1998dynamic,engwerda2005lq}.  For example, Nash game requires that no player can improve its objective through a unilateral deviation while the other players' strategies remain fixed ~\cite{nash1951non}. While Stackelberg game instead describes a hierarchical interaction in which a leader selects its strategy while anticipating the best response of a follower ~\cite{basar1979closed}.

Classical pursuit--evasion mainly study the competition between a pursuer and an evader, including capture conditions, optimal trajectories, and winning
regions~\cite{isaacs1999differential}. In zero-sum settings, optimal strategies and winning regions are often
characterized using Hamilton--Jacobi--Isaacs equations and reachability analysis ~\cite{mitchell2005time,bansal2017hamilton}. In another target-protection scenarios, the evader may also receive assistance from a defender. The outcome then depends on whether the pursuer can capture the evader before the defender intervenes, introducing a cooperative relationship alongside the existing adversarial interaction.

The resulting three-agent pursuit--evasion--defense problem comprises an attacker (the pursuer), a target (the evader), and a defender that protects the target by intercepting or delaying the attacker.  Fisac and Sastry~\cite{fisac2015pursuit} analyzed this interaction using a double-obstacle Hamilton--Jacobi--Isaacs variational inequality. Garcia et al.~\cite{garcia2018design,garcia2021complete} modeled the target and defender as a cooperative team against the attacker and derived state-feedback saddle-point
strategies for the resulting zero-sum differential game. Liang et al.~\cite{liang2019differential} adopted a similar team-based formulation and constructed a barrier separating the winning regions of the two sides. Extensions to multiple attackers and defenders have employed time-optimal interception and assignment methods \cite{coon2017control}. However, despite their different solution
methods, these studies generally treat the target and defender as a cooperative team without modeling strategic asymmetry or a leader--follower hierarchy within the
protective group.

Hierarchical and simultaneous interactions have been combined more broadly in mixed-leadership dynamic games. Ba{\c{s}}ar et al.~\cite{basar2010mixed} introduced an
open-loop formulation in which each player leads in some decisions and follows in others. Linear--quadratic and feedback variants were subsequently developed using
coupled Riccati and Hamilton--Jacobi--Bellman equations, respectively
\cite{bensoussan2013linear,bensoussan2019feedback}.
 These studies demonstrate that hierarchical and simultaneous interactions can coexist, but they don't consider a fixed leader--follower pair with distinct objectives
competing against an independently optimizing third agent.

We study a leader--defender--attacker game in which the leader carries out its mission with protection provided by the defender. The leader and defender form a group with a shared protective task while retaining distinct individual objectives. The attacker pursues its own objective against this group, giving rise to cooperation within the group and competition between the group and the attacker. To capture these interactions, we propose a nested equilibrium formulation. For each fixed attacker strategy, the leader--defender subsystem forms a Stackelberg equilibrium in which the leader anticipates the defender's best response. The attacker satisfies a Nash-type best-response condition, evaluating unilateral deviations while holding the leader's and defender's strategies fixed. A nested equilibrium requires these internal and external equilibrium conditions to hold simultaneously.
\subsection*{Contributions}

\begin{itemize}
    \item We formulate a finite-horizon leader--defender--attacker game with distinct  objectives and introduce a nested equilibrium that to characterize the internal leader--follower hierarchy and external competition with the attacker.

    \item We establish sufficient conditions for the existence and uniqueness of the nested equilibrium and characterize the resulting winning regions in terms of the initial joint state.

    \item We develop a receding-horizon algorithm to solve for nested equilibrium and validate its effectiveness through numerical simulations.

\end{itemize}

\section{Problem Formulation}

\begin{figure}
    \centering
    \includegraphics[width=\linewidth]{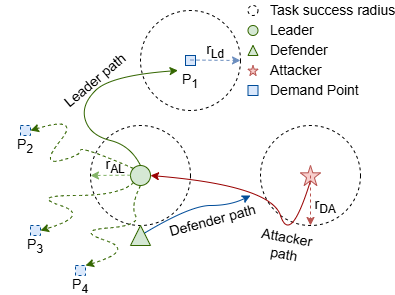}
    \caption{Leader-Defender-Attacker Problem Setting. The first agent to enter its corresponding
task-success region, defined by the radius $r_\alpha$, is declared the winner.}
    \label{fig:1}
\end{figure}

We consider a three-agent game including a leader (L), a defender (D), and an attacker (A) as shown in Figure \ref{fig:1}. The attacker aims to attack the leader by approaching the leader within a certain distance, while the leader needs to reach demand points to get resources while avoiding the attack. The defender is with the leader and his objective is to intercept the attacker while stay in the leader's vicinity.
Therefore, the leader and the defender work as a group to jointly play against the attacker. Within the group, the leader and the defender follows a hierarchical interaction, i.e., the leader decides and moves first while the defender takes his move after observing the leader's action.
%
We denote $x_t^i \in \R^n$ and $u^i_t \in \R^{m_i}$ as the agent $i$'s state and input at time $t$, $i \in \{L,D,A\}$. The state $x_i$ includes the position and velocity of each agent. Also let $x^{\mathrm{d}}$ be the state including the position of the demand point.
We assume linear dynamics for all agents, where we have
\begin{equation}
    \label{eq:dynamics}
    x^i_{t+1} = A^i x^i_t + B^i u^i_t, \quad i \in \{L,D,A\}.
\end{equation}
The interactions among the three agents can be characterized by the following cost functions:

\begin{equation}
\label{eq:J_leader}
\begin{aligned}
J^L =& 
\sum_{t=0}^{T-1}\gamma^t
\Big(
\|x_t^L-x^\mathrm{d}\|_{Q^{L_1}}^2
+\|u_t^L\|_{R^L}^2
-
\|x_t^L-x_t^A\|_{Q^{L_2}}^2\Big)\\
&+
\gamma^T
\Big(
\|x_T^L-x^\mathrm{d}\|_{Q^{L_1}}^2
-\|x_T^L-x_T^A\|_{Q^{L_2}}^2
\Big),
\end{aligned}
\end{equation}
\begin{equation}
\label{eq:J_defender}
\begin{aligned}
J^D =& \sum_{t=0}^{T-1}\gamma^t
\Big(
\|x_t^D-x_t^A\|_{Q^{D_1}}^2
+\|u_t^D\|_{R^D}^2
+\|x_t^L-x_t^D\|_{Q^{D_2}}^2
\Big)\\
&+\gamma^T
\Big(
\|x_T^D-x_T^A\|_{Q^{D_1}}^2
+\|x_T^L-x_T^D\|_{Q^{D_2}}^2
\Big),
\end{aligned}
\end{equation}
\begin{equation}
\label{eq:J_attacker}
\begin{aligned}
J^A =& \sum_{t=0}^{T-1}\gamma^t
\Big(
\|x_t^A-x_t^L\|_{Q^{A_1}}^2
+\|u_t^A\|_{R^A}^2
-\|x_t^A-x_t^D\|_{Q^{A_2}}^2
\Big)\\
&+\gamma^T
\Big(
\|x_T^A-x_T^L\|_{Q^{A_1}}^2
-\|x_T^A-x_T^D\|_{Q^{A_2}}^2
\Big).
\end{aligned}
\end{equation}
where $\gamma \in (0,1)$ is the discounted factor, $Q^{k}, R^i$ are penalty matrices, $k \in \{L_1, L_2, D_1, D_2, A_1, A_2\}$, $i \in \{L,D,A\}$.

The types of interactions in three agents are heterogeneous. Within the leader--defender group, their hierarchical interactions can be characterized as a Stackelberg game, with the attacker as an external player. Outside the group, they are viewed as one independent decision maker to interact with the attacker, leading to a normal-form dynamic game. Therefore, the overall problem can be formulated as a \emph{nested} game with two levels: the inner level determines the leader--defender Stackelberg strategy, and the outer level seeks the best-response strategy of the attacker. Such game couples the Stackelberg and Nash equilibria solutions, which requires a new equilibrium solution concept to characterize. This leads to the following problem.

\begin{problem}
\label{prob:nested_equilibrium}
Given the dynamics \eqref{eq:dynamics} and the objectives
\eqref{eq:J_leader}--\eqref{eq:J_attacker}, what are the optimal strategies for all three agents? Is there a winning region for each agent?
\end{problem}

Let $\pi^i$ denote the strategy of agent $i \in \{L,D,A\}$, and let $\Pi^i$ denote the corresponding admissible strategy set. Then we introduce the following equilibrium solution concept.

\begin{definition}[ Nested Equilibrium]
\label{def:nested}
We use $\operatorname{SE}$ to denote the Stackelberg equilibrium operator. For a fixed attacker strategy $\pi^A \in \Pi^A$, a strategy pair $(\pi^L, \pi^D) \in \operatorname{SE}(\pi^A)$ satisfies
\begin{equation*}
    J^D(\pi^{D*},\pi^L,\pi^A) \leq J^D(\pi^D,\pi^L,\pi^A), \quad \forall \, \pi^D \in \Pi^D,
\end{equation*}
and $\forall \, \pi^L \in \Pi^L$
\begin{equation*}
    J^L(\pi^{L*}, \pi^{D*} (\pi^{L*}), \pi^A) \leq J^L(\pi^L, \pi^{D*}(\pi^L), \pi^A).
\end{equation*}
A strategy triple $(\pi^{L*},\pi^{D*},\pi^{A*})$ constitutes a \emph{ nested equilibrium} if
\begin{equation*}
    J^A(\pi^{A*},\pi^{L*},\pi^{D*}) \leq   J^A(\pi^{A},\pi^{L*},\pi^{D*}), \quad \forall \, \pi^A \in \Pi^A,
\end{equation*}
and
\begin{equation*}
    (\pi^{L*},\pi^{D*}) \in \operatorname{SE}(\pi^{A*}).
\end{equation*}
\end{definition}

\section{Equilibrium and Winning Region Analysis} \label{sec:analysis}

To simplify analysis, we first perform some reformulations.
Based on \eqref{eq:dynamics}, we write the joint dynamics as
\begin{equation*}
    X_{t+1} = A X_t + B U_t
\end{equation*}
where $X_t = \begin{bmatrix} x_t^L, x_t^D, x_t^A \end{bmatrix} \in\mathbb{R}^{3n}$, $U_t = \begin{bmatrix} u_t^L, u_t^D, u_t^A \end{bmatrix} \in\mathbb{R}^{m_L+m_D+m_A}$. The aggregated system matrices are $A = \operatorname{diag}(A_t^L, A_t^D, A_t^A) \in \mathbb{R}^{3n\times 3n}$ and $B= \operatorname{diag}(B_t^L, B_t^D, B_t^A) \in \mathbb{R}^{3n\times (m_L+m_D+m_A) }$. Here, $\operatorname{diag}(\cdot)$ places the input matrices on the diagonal position and concatenate them.
%
Next, we reformulate the objectives in Eq. \eqref{eq:J_leader}--\eqref{eq:J_attacker} into an equivalent quadratic form. 
To simplify analysis, we write the objectives \eqref{eq:J_leader}--\eqref{eq:J_attacker} into compact forms.
Specifically, for each agent $i$, we express the stage cost as
\begin{equation*}
    \ell^i_{{t}}(X_t,{u}_t^i) = X_t^\top \widetilde Q^i X_t + (\widetilde q^i)^\top X_t + c^i + (u_t^i)^\top R^i u_t^i
\end{equation*}
where $\widetilde Q^i$, $\widetilde q^i$, and $c^i$ are 
given by
\begin{equation*}
\label{eq:quadratic_cost_parameters}
\left\{
\begin{aligned}
\widetilde{Q}^L
&=
\begin{bmatrix}
Q^{L_1}-Q^{L_2} & 0 & Q^{L_2}\\
0&0&0\\
Q^{L_2}&0&-Q^{L_2}
\end{bmatrix}
\\[3mm]
\widetilde{Q}^D
&=
\begin{bmatrix}
Q^{D_2}&-Q^{D_2}&0\\
-Q^{D_2}&Q^{D_1}+Q^{D_2}&-Q^{D_1}\\
0&-Q^{D_1}&Q^{D_1}
\end{bmatrix}
\\[3mm]
\widetilde{Q}^A
&=
\begin{bmatrix}
Q^{A_1}&0&-Q^{A_1}\\
0&-Q^{A_2}&Q^{A_2}\\
-Q^{A_1}&Q^{A_2}&Q^{A_1}-Q^{A_2}
\end{bmatrix}
\end{aligned}
\right.,
\end{equation*}
%
\begin{equation*}
\label{eq:linear_cost_parameters}
\widetilde q^L
=
\begin{bmatrix}
-2Q^{L_1}x^d\\
0\\
0
\end{bmatrix},
\quad
\widetilde q^D=0,
\quad
\widetilde q^A=0,
\end{equation*}
and 
\begin{equation*}
\label{eq:constant_cost_parameters}
c^L = (x^d)^\top Q^{L_1}x^d, \quad c^D=0, \quad c^A = 0.
\end{equation*}
The terminal cost can be written as
\begin{equation*}
\label{eq:terminal_cost}
\phi_T^i(X_T)
=
X_T^\top \widetilde Q^i X_T
+
(\widetilde q^i)^\top X_T
+
c^i.
\end{equation*}

For each agent $i\in\{L,D,A\}$, let $V_t^i(X_t)$ denote the cost-to-go from time $t$ to the terminal time $T$:
\begin{equation*}
V_t^i(X_t)
=
\sum_{\tau=t}^{T-1}
\gamma^{\tau-t}
\ell_{{\tau}}^i(X_\tau,{u}_\tau^i)
+
\gamma^{T-t}\phi_T^i(X_T).
\end{equation*}
Under the linear--quadratic structure, the cost-to-go function is
quadratic-affine in the joint state:
\begin{equation}
\label{eq:value_function}
V_t^i(X_t) = X_t^\top P_t^iX_t + 2(s_t^i)^\top X_t + c_t^{V,i}.
\end{equation}
At the terminal time,
\begin{equation}
\label{eq:terminal_conditions}
P_T^i = \widetilde Q^i,
\quad
s_T^i=\frac{1}{2}\widetilde q^i,
\quad
c_T^{V,i}=c^i.
\end{equation}
The matrices $P_t^i$ and vectors $s_t^i$ are computed backward from $t=T-1$ to $t=0$ using dynamic programming. 

To simplify the notation and facilitate the subsequent derivation, we introduce the following feedback gains.
For the leader, define
\begin{equation*}
\label{eq:L_gains_R}
\left\{
\begin{aligned}
\widehat K_{x,t}^L
&=
\gamma(\widehat H_t^L)^{-1}
(\widehat B_t^L)^\top
P_{t+1}^L\bar A_t,
\\
\widehat K_{A,t}^L
&=
\gamma(\widehat H_t^L)^{-1}
(\widehat B_t^L)^\top
P_{t+1}^L\bar B_t^A,
\\
\widehat k_t^L
&=
\gamma(\widehat H_t^L)^{-1}
(\widehat B_t^L)^\top
\left(
P_{t+1}^L\bar b_t+s_{t+1}^L
\right).
\end{aligned}
\right.
\end{equation*}

For the defender and attacker, indexed by $i\in\{D,A\}$, the corresponding feedback gains are defined as
\begin{equation*}
\label{eq:DA_gains_R}
\left\{
\begin{aligned}
K_{x,t}^i
&=
\gamma(H_t^i)^{-1}
(B^i)^\top P_{t+1}^i A,
\\
K_{j,t}^i
&=
\gamma(H_t^i)^{-1}
(B^i)^\top P_{t+1}^i B^j,
\qquad
j\in\{L,D,A\}\setminus\{i\},
\\
k_t^i
&=
\gamma(H_t^i)^{-1}
(B^i)^\top s_{t+1}^i.
\end{aligned}
\right.
\end{equation*}

where 
\begin{equation*}
\label{eq:reduced_matrices_R}
\left\{
\begin{aligned}
\bar A_t
&=
A-B^DK_{x,t}^D,
\\
\widehat B_t^L
&=
B^L-B^DK_{L,t}^D,
\\
\bar B_t^A
&=
B^A-B^DK_{A,t}^D,
\\
\bar b_t
&=
-B^Dk_t^D.
\end{aligned}
\right.
\end{equation*}
and 
\begin{equation*}
\label{eq:hessian_matrices}
\left\{
\begin{aligned}
H_t^D
&=
R^D+
\gamma(B^D)^\top P_{t+1}^D B^D,
\\[2mm]
\widehat H_t^L
&=
R^L+
\gamma(\widehat B_t^L)^\top
P_{t+1}^L
\widehat B_t^L,
\\[2mm]
H_t^A
&=
R^A+
\gamma(B^A)^\top P_{t+1}^A B^A.
\end{aligned}
\right.
\end{equation*}
Then we arrive at the following theorem.

\begin{theorem}[Existence of Affine Feedback Equilibrium]
\label{theorem:closed_form_equilibrium}

Under the linear dynamics, the first-order optimality conditions characterizing the nested equilibrium can be expressed as
\begin{equation*}
\mathcal M_tU_t
+
\mathcal K_{x,t}X_t
+
\kappa_t
=0 .
\end{equation*}
where
\begin{equation*}
\mathcal M_t
=
\begin{bmatrix}
I & 0 & \widehat K_{A,t}^L\\
K_{L,t}^D & I & K_{A,t}^D\\
K_{L,t}^A & K_{D,t}^A & I
\end{bmatrix}.
\end{equation*}
and 
\begin{equation*}
\label{eq:state_affine_terms}
\mathcal K_{x,t}
=
\begin{bmatrix}
\widehat K_{x,t}^L\\
K_{x,t}^D\\
K_{x,t}^A
\end{bmatrix},
\qquad
\kappa_t
=
\begin{bmatrix}
\widehat k_t^L\\
k_t^D\\
k_t^A
\end{bmatrix}.
\end{equation*}
If $\mathcal M_t$ is nonsingular, $\widehat H_t^L\succ0$, $H_t^D\succ0$, and $H_t^A\succ0$, the game admits an affine feedback equilibrium:
\begin{equation}
\label{eq:1}
U_t^* = K_t^*X_t+k_t^* ,
\end{equation}
where $K_t^*=-\mathcal M_t^{-1}\mathcal K_{x,t},$ and $k_t^*=-\mathcal M_t^{-1}\kappa_t .$
\end{theorem}

\begin{proof}
\paragraph{Defender response}

The continuation value function $V_{t+1}^D(X_{t+1})$ is
given by Eq. \eqref{eq:value_function}. For fixed $(U_t^L,U_t^A)$, the defender
solves
\begin{equation*}
\label{eq:bellman_defender_proof}
\min_{U_t^D}
\left\{
\ell^D(X_t,U_t^D)
+
\gamma V_{t+1}^D(X_{t+1})
\right\}.
\end{equation*}
The first-order optimality condition is
\begin{align*}
0
={}&
2
\left[
R^D+
\gamma(B^D)^\top P_{t+1}^DB^D
\right]U_t^D
\nonumber\\
&+
2\gamma(B^D)^\top P_{t+1}^DA X_t
\nonumber\\
&+
2\gamma(B^D)^\top P_{t+1}^DB^L U_t^L
\nonumber\\
&+
2\gamma(B^D)^\top P_{t+1}^DB^A U_t^A
+
2\gamma(B^D)^\top s_{t+1}^D.
\end{align*}
If $H_t^D\succ0$, the defender's optimization problem is strictly convex in $U_t^D$ and its unique best response is
\begin{equation}
\label{eq:UD_response_R}
U_t^D
=
-K_{x,t}^DX_t
-K_{L,t}^DU_t^L
-K_{A,t}^DU_t^A
-k_t^D.
\end{equation}

\paragraph{Leader response.}

The leader anticipates the defender's best response
\eqref{eq:UD_response_R}. Substituting this response into the system
dynamics gives
\begin{equation*}
\label{eq:reduced_dynamics_R}
X_{t+1}
=
\bar A_tX_t
+
\widehat B_t^LU_t^L
+
\bar B_t^AU_t^A
+
\bar b_t,
\end{equation*}
In particular,
\begin{equation*}
\label{eq:effective_leader_input_R}
\frac{\partial X_{t+1}}{\partial U_t^L}
=
\widehat B_t^L
=
B^L-B^DK_{L,t}^D.
\end{equation*}

The leader's Bellman problem is therefore
\begin{equation*}
\label{eq:bellman_leader_R}
\min_{U_t^L}
\left\{
\ell^L(X_t,U_t^L)
+
\gamma V_{t+1}^L(X_{t+1})
\right\}.
\end{equation*}
Its first-order optimality condition is
\begin{align*}
0
={}&
2
\left[
R^L
+
\gamma(\widehat B_t^L)^\top
P_{t+1}^L\widehat B_t^L
\right]U_t^L
\nonumber\\
&+
2\gamma(\widehat B_t^L)^\top
P_{t+1}^L\bar A_tX_t
\nonumber\\
&+
2\gamma(\widehat B_t^L)^\top
P_{t+1}^L\bar B_t^AU_t^A
\nonumber\\
&+
2\gamma(\widehat B_t^L)^\top
\left(
P_{t+1}^L\bar b_t+s_{t+1}^L
\right).
\end{align*}
If $\widehat H_t^L\succ0$, the leader's unique Stackelberg response is
\begin{equation}
\label{eq:UL_response_R}
U_t^L
=
-\widehat K_{x,t}^LX_t
-\widehat K_{A,t}^LU_t^A
-\widehat k_t^L,
\end{equation}

\paragraph{Attacker response.}

The first-order optimality condition is
\begin{align*}
0
={}&
2
\left[
R^A+
\gamma(B^A)^\top P_{t+1}^AB^A
\right]U_t^A
\nonumber\\
&+
2\gamma(B^A)^\top P_{t+1}^AA X_t
\nonumber\\
&+
2\gamma(B^A)^\top P_{t+1}^AB^L U_t^L
\nonumber\\
&+
2\gamma(B^A)^\top P_{t+1}^AB^D U_t^D
+
2\gamma(B^A)^\top s_{t+1}^A.
\end{align*}
If $H_t^A\succ0$, the attacker's unique best response is
\begin{equation}
\label{eq:UA_response_R}
U_t^A
=
-K_{x,t}^AX_t
-K_{L,t}^AU_t^L
-K_{D,t}^AU_t^D
-k_t^A,
\end{equation}

\paragraph{Nested equilibrium.}

Collecting the defender's best response
\eqref{eq:UD_response_R}, the leader's Stackelberg response
\eqref{eq:UL_response_R}, and the attacker's best response
\eqref{eq:UA_response_R}, we obtain the coupled linear system
 \begin{equation*}
\label{eq:block_system_R}
\mathcal M_tU_t
+
\mathcal K_{x,t}X_t
+
\kappa_t
=
0,
\end{equation*}

If $\mathcal M_t$ is nonsingular, the coupled equilibrium controls
are uniquely determined by Eq. \eqref{eq:1}.
\end{proof}

Substituting the affine equilibrium control Eq.
\eqref{eq:1} into the closed-loop dynamics
gives
\begin{equation}
\label{eq:closed_loop_dynamics_R}
X_{t+1}
=
A_t^{\mathrm{cl}}X_t+b_t^{\mathrm{cl}},
\end{equation}
where
\begin{equation*}
\label{eq:closed_loop_matrices_R2}
A_t^{\mathrm{cl}} =
A+B K_t^*,
\quad 
b_t^{\mathrm{cl}}
=
B k_t^*.
\end{equation*}
Substituting the equilibrium control and the
closed-loop dynamics into the Bellman recursion yields
\begin{equation*}
P_t^i
=
\widetilde Q^i
+
(K_t^{i*})^\top R^iK_t^{i*}
+
\gamma
(A_t^{\mathrm{cl}})^\top
P_{t+1}^i
A_t^{\mathrm{cl}},
\end{equation*}
and
\begin{equation*}
s_t^i
=
\frac{1}{2}\widetilde q^i
+
(K_t^{i*})^\top R^ik_t^{i*}
+
\gamma
(A_t^{\mathrm{cl}})^\top
\left(
P_{t+1}^ib_t^{\mathrm{cl}}
+s_{t+1}^i
\right).
\end{equation*}
The constant term satisfies
\begin{equation*}
\begin{split}
c_t^{V,i} = c^i
& + (k_t^{i*})^\top R^ik_t^{i*} \\
& + \gamma
\left[
(b_t^{\mathrm{cl}})^\top
P_{t+1}^ib_t^{\mathrm{cl}}
+
2(s_{t+1}^i)^\top b_t^{\mathrm{cl}}
+
c_{t+1}^{V,i}
\right].
\end{split}
\end{equation*}

Therefore, starting from the terminal conditions Eq.
\eqref{eq:terminal_conditions}, the continuation-value parameters are
updated backward from $t=T-1$ to $t=0$.

Furthermore, to analyze the influence of the initial conditions on the final outcome, define the relative positions for $\alpha\in \{AL,DA,Ld\}$ as
\begin{equation}
\label{eq:2}
\rho_t^{AL}=p_t^A-p_t^L,
\rho_t^{DA}=p_t^D-p_t^A,
\rho_t^{Ld}=p_t^L-p^d,
\end{equation} 

with $d_t^\alpha=\|\rho_t^\alpha\|_2$, where $p_t^i$ denotes the position of agent $i$, and $AL$, $DA$, and $Ld$ represent attacker capture, defender interception, and leader arrival, respectively.

\begin{definition}[Event Completion Time]
\label{def:event_time}

The completion time in Definition~\ref{def:event_time}
therefore satisfies
\begin{equation*}
T^\alpha(X_0)
=
\inf\{t\in\mathbb Z_{\geq0}:\left\|d_t^\alpha(X_0)\right\|^2\leq r_{\alpha}^2\},
\end{equation*}

If the event never occurs, $T^\alpha=+\infty$.
The winning region of agent $i\in\{L,D,A\}$ is the set of initial
states from which agent $i$ is the first to achieve its corresponding
terminal event under the equilibrium policy. It is denoted by
$\mathcal W^i$.
\end{definition}

Under the affine equilibrium policy, let $\mathcal E=\{AL,DA,Ld\}$ denote the event index set.
For each $\alpha\in\mathcal E$, define
\begin{equation*}
\left\{
\begin{aligned}
F_0 &= I,
&
g_0 &= 0,\\
F_{t+1} &= A_t^{\rm cl}F_t,
&
g_{t+1} &= A_t^{\rm cl}g_t+b_t^{\rm cl}.
\end{aligned}
\right.
\end{equation*}

For each event $\alpha$, define
\begin{equation*}
\label{eq:event_mapping}
\Psi_t^\alpha=C^\alpha F_t,
\qquad
\eta_t^\alpha=C^\alpha g_t+c^\alpha,
\end{equation*}
where
\begin{equation*}
\label{eq:event_matrices}
\left\{
\begin{aligned}
C^{AL} &= F_p(E^A-E^L),
&\qquad c^{AL} &= 0,\\
C^{DA} &= F_p(E^D-E^A),
& c^{DA} &= 0,\\
C^{Ld} &= F_pE^L,
& c^{Ld} &= -p^d.
\end{aligned}
\right.
\end{equation*}

\begin{theorem}[Initial-State-Dependent Winner Regions]
\label{thm:winner_region}

Under the affine equilibrium policy and the closed-loop dynamics
\eqref{eq:closed_loop_dynamics_R}, for each
$\alpha\in\{AL,DA,Ld\}$, the event distance satisfies
\begin{equation*}
d_t^\alpha(X_0)
=
\left\|
\Psi_t^\alpha X_0+\eta_t^\alpha
\right\|.
\end{equation*}
Hence, the completion time of event $\alpha$ is determined by the
initial state as
$T^\alpha(X_0)$.

Assume that the game terminates when the first terminal event occurs and
that the event times are distinct. Then the winning regions are
\begin{equation*}
\begin{aligned}
\mathcal W^A
&=
\left\{
X_0:
T^{AL}(X_0)
<
\min\left\{
T^{DA}(X_0),T^{Ld}(X_0)
\right\}
\right\},
\\
\mathcal W^D
&=
\left\{
X_0:
T^{DA}(X_0)
<
\min\left\{
T^{AL}(X_0),T^{Ld}(X_0)
\right\}
\right\},
\\
\mathcal W^L
&=
\left\{
X_0:
T^{Ld}(X_0)
<
\min\left\{
T^{AL}(X_0),T^{DA}(X_0)
\right\}
\right\}.
\end{aligned}
\end{equation*}

\end{theorem}

\begin{proof}
Under the closed-loop dynamics  of Eq. (\ref{eq:closed_loop_dynamics_R}).
Using $F_0=I$, $g_0=0$, and the recursions
\begin{equation*}
    F_{t+1} = A_t^{\rm cl}F_t,
    \qquad
    g_{t+1} = A_t^{\rm cl}g_t+b_t^{\rm cl},
\end{equation*}
a direct induction gives
\begin{equation*}
\label{eq:state_initial_mapping}
    X_t = F_tX_0+g_t,
    \qquad
    t \in \mathbb Z_{\ge0}.
\end{equation*}
Based on Eq. (\ref{eq:2}), the corresponding relative position can be written as
\begin{equation*}
    d_t^\alpha 
    = C^\alpha(F_tX_0+g_t)+c^\alpha
    = \Psi_t^\alpha X_0+\eta_t^\alpha.
\end{equation*}
Consequently,
\begin{equation*}
    (d_t^\alpha)^2 =
    \left\| \Psi_t^\alpha X_0+\eta_t^\alpha \right\|^2.
\end{equation*}
Therefore, the terminal condition $d_t^\alpha\le r_\alpha$ is
equivalent to
\begin{equation*}
q_t^\alpha(X_0)
=
\left\|
\Psi_t^\alpha X_0+\eta_t^\alpha
\right\|^2-r_\alpha^2
\le0.
\end{equation*}
Since the game evolves in discrete time, the completion time of event
$\alpha$ is thus
\begin{equation*}
    T^\alpha(X_0) =
    \inf \left\{ t\in\mathbb Z_{\ge0}: q_t^\alpha(X_0) \leq 0 \right\}.
\end{equation*}

The game terminates at the first event encountered along the closed-loop equilibrium trajectory. Hence, the attacker wins precisely when
\begin{equation*}
    T^{AL} < \min\{T^{DA},T^{Ld}\},
\end{equation*}
the defender wins precisely when
\begin{equation*}
    T^{DA} < \min\{T^{AL},T^{Ld}\},
\end{equation*}
and the leader wins precisely when
\begin{equation*}
    T^{Ld} < \min\{T^{AL},T^{DA}\}.
\end{equation*}
\end{proof}

\begin{remark}[Planning objective and execution termination]
The costs in \eqref{eq:J_leader}--\eqref{eq:J_attacker}
define a fixed-horizon quadratic game. Capture, interception,
and arrival events do not truncate these planning costs.
Instead, they serve as stopping criteria during execution
and as outcome measures for the resulting trajectories.
Accordingly, the equilibrium characterization established
above concerns the fixed-horizon game; it does not establish
an equilibrium of a game whose costs terminate at the
first event.
\end{remark}

\section{Algorithm Design}
\label{sec:algorithm}

The preceding analysis characterizes the finite-horizon equilibrium for a given current state. To apply this equilibrium over the entire game evolution, we embed the equilibrium solver into a receding-horizon algorithm. The resulting procedure repeatedly predicts the future game, computes the corresponding nested equilibrium, executes only the first control action, and updates the prediction using the newly observed state.

At the beginning of sampling instant $t$, the current joint state $X_t$ is measured and treated as the initial condition of the prediction problem. Let $X_{t+k|t}$ denote the predicted joint state $k$ steps ahead based on the information available at time $t$, where
\begin{equation*}
    X_{t|t}=X_t, \qquad k=0,\ldots,H.
\end{equation*}
Similarly, let $U_{t+k|t}^*$ denote the equilibrium control predicted for the corresponding stage. The predicted trajectory is generated according to
\begin{equation*}
\label{eq:prediction_dynamics}
    X_{t+k+1|t} = A X_{t+k|t} + BU_{t+k|t}^*,
    \quad 
    k=0,\ldots,H-1.
\end{equation*}

For each prediction step, the finite-horizon equilibrium conditions derived previously are evaluated using the continuation values associated with the remaining prediction horizon. This produces a sequence of equilibrium feedback laws
\begin{equation}
\label{eq:predicted_policy_sequence}
    U_{t+k|t}^* = K_{t+k|t}^*X_{t+k|t} + k_{t+k|t}^*,
    \quad
    k=0,\ldots,H-1,
\end{equation}
where the gains and affine terms are obtained from the backward equilibrium recursion. 
We use the first control action is implemented:
\begin{equation*}
\label{eq:receding_control}
    U_t = U_{t|t}^*.
\end{equation*}
After the system evolves, the next state $X_{t+1}$ is observed and a new prediction problem is initialized according to
\begin{equation*}
    X_{t+1|t+1}=X_{t+1}.
\end{equation*}
Hence, the prediction window moves forward by one step,
\begin{equation*}
\label{eq:horizon_shift}
    [t,t+H]
    \;\longrightarrow\;
    [t+1,t+H+1],
\end{equation*}
and the finite-horizon equilibrium is recomputed using the updated state. This operation continues until a terminal event is detected. 
The overall closed-loop procedure is summarized in Algorithm~\ref{alg:rolling}.

\begin{algorithm}[t]
\caption{Receding-Horizon Nested Equilibrium Control}
\label{alg:rolling}
\begin{algorithmic}[1]
\Require Initial state $X_0$, prediction horizon $H$, and maximum
number of steps $N_{\mathrm{total}}$.

\For{$t=0,\ldots,N_{\mathrm{total}}-1$}

    \State Measure the current joint state $X_t$ and set
    $X_{t|t}\leftarrow X_t$.

    \State Initialize the terminal value parameters at $t+H$.

    \State Solve the finite-horizon nested equilibrium problem
    backward over the prediction window
    $\tau=t+H-1,\ldots,t$.

    \State Generate the predicted equilibrium trajectory
    $\{X_{\tau|t},U_{\tau|t}^*\}_{\tau=t}^{t+H-1}$.

    \State Apply only the first control action
    $U_t\leftarrow U_{t|t}^*$.

    \State Observe the resulting state $X_{t+1}$.

    \If{a terminal event occurs}
        \State \textbf{break}
    \EndIf

\EndFor
\end{algorithmic}
\end{algorithm}

\section{Numerical Experiments}
\label{sec:simulation}

We evaluate the proposed algorithm in a two-dimensional leader--defender--attacker environment. Its performance is assessed through a prediction-horizon sensitivity analysis and comparative experiments over 50 randomly generated initial configurations. The simulation settings are
described below.

\subsection{Simulation Setup}

The initial positions of the three agents are sampled
independently and uniformly from
$[0,10]\,\mathrm{m}\times[0,10]\,\mathrm{m}$, and their
initial velocities are set to zero. Each simulation runs
for at most $N_{\mathrm{total}}=200$ steps.

\paragraph{Dynamics and parameters.}
Each agent $i\in\{L,D,A\}$ follows the discrete-time
double-integrator dynamics \eqref{eq:dynamics}, with
\begin{equation*}
    A^i=
    \begin{bmatrix}
        I_2 & \Delta t I_2\\
        0_2 & I_2
    \end{bmatrix},
    \quad
    B^i=
    \begin{bmatrix}
        \frac{1}{2}\Delta t^2 I_2\\
        \Delta t I_2
    \end{bmatrix},
    \quad
    \Delta t=0.1\,\mathrm{s}.
\end{equation*}
The state and control input are
\begin{equation*}
    x_t^i=
    \begin{bmatrix}
        (p_t^i)^\top & (v_t^i)^\top
    \end{bmatrix}^{\top}\in\mathbb R^4,
    \quad
    u_t^i=
    \begin{bmatrix}
        a_{x,t}^i & a_{y,t}^i
    \end{bmatrix}^{\top}\in\mathbb R^2,
\end{equation*}
where $p_t^i$, $v_t^i$, and $u_t^i$ denote the position, velocity, and acceleration of agent $i$, respectively.

At each time step, the strategy sequence in
\eqref{eq:predicted_policy_sequence} is computed over a prediction horizon $H$, and only its first control input is
applied. The cost parameters are
\begin{equation*}
    Q^{i_1}=2I_4,\quad
    Q^{i_2}=I_4,\quad
    R^i=0.1I_2,
    \quad i\in\{L,D,A\},
\end{equation*}
with discount factor $\gamma=0.99$.

\paragraph{Terminal events.}
Leader arrival, attacker capture, and defender interception occur when
\begin{equation*}
    d_t^{Ld}\leq r_{Ld},\quad
    d_t^{AL}\leq r_{AL},\quad
    d_t^{DA}\leq r_{DA},
\end{equation*}
respectively, where $r_{Ld}=r_{AL}=r_{DA}=0.3\mathrm{m}$.
Execution stops at the first terminal event. Simultaneous first events are classified as ties, and trials with no event within $N_{\mathrm{total}}$ steps are classified as
no event. The terminal events stop execution but do not truncate the fixed-horizon planning costs.

\subsection{Different Prediction Horizon}
\label{subsec:sensitivity}

We evaluate sensitivity to the prediction horizon using $H\in\{1,5,10,15,20,25,30\}$. A fixed set of 50 randomly sampled initial configurations is reused for every horizon to ensure a consistent comparison. Each simulation runs for at most 300 steps; ``No event'' indicates that no terminal event occurs within this limit.

\begin{table}[H]
\centering
\caption{Outcomes over 50 random initial configurations.
Each entry gives the number of runs and its percentage.}
\label{tab:random_outcomes}
\setlength{\tabcolsep}{3pt}
\resizebox{\columnwidth}{!}{%
\begin{tabular}{@{}lccccccc@{}}
\toprule
Outcome
& $H=1$ & $H=5$ & $H=10$ & $H=15$
& $H=20$ & $H=25$ & $H=30$ \\
\midrule
Leader wins
& 3 (6\%) & 16 (32\%) & 16 (32\%) & 13 (26\%)
& 14 (28\%) & 15 (30\%) & 15 (30\%) \\
Defender wins
& 1 (2\%) & 23 (46\%) & 16 (32\%) & 12 (24\%)
& 13 (26\%) & 12 (24\%) & 12 (24\%) \\
Attacker wins
& 0 (0\%) & 11 (22\%) & 18 (36\%) & 25 (50\%)
& 23 (46\%) & 23 (46\%) & 23 (46\%) \\
No event
& 46 (92\%) & 0 (0\%) & 0 (0\%) & 0 (0\%)
& 0 (0\%) & 0 (0\%) & 0 (0\%) \\
\midrule
Total
& 50 (100\%) & 50 (100\%) & 50 (100\%) & 50 (100\%)
& 50 (100\%) & 50 (100\%) & 50 (100\%) \\
\bottomrule
\end{tabular}%
}
\end{table}

As shown in Table~\ref{tab:random_outcomes}, $H=1$
produces no terminal event in $92\%$ of the trials,
indicating that one-step prediction is too myopic under
the considered setting. All trials terminate for
$H\geq5$. The defender has the highest win rate at $H=5$,
whereas the outcomes are more balanced at $H=10$. For
$H\geq15$, the attacker has the highest win rate, reaching
$50\%$ at $H=15$ and remaining at $46\%$ for
$H\in\{20,25,30\}$.

The outcome distribution changes only slightly beyond
$H=20$ and is identical for $H=25$ and $H=30$, suggesting
that the effect of the prediction horizon saturates at
the level of aggregate outcome frequencies. This does not
imply identical trajectories or control inputs. These
frequencies are empirical results for the sampled initial
conditions and should not be interpreted as formal
measures of the winning regions.

Figure~\ref{fig:horizon_outcome_examples} compares the
trajectories generated from the same initial configuration
under different prediction horizons. In this example, the
prediction horizon affects both the terminal outcome and
the termination time. No event occurs for $H=1$; the
defender wins for $H=5$ and $H=10$; and the leader wins
for $H\geq15$. Among the completed runs, the termination
time decreases as the prediction horizon increases.

\begin{figure}[t]
\centering
\includegraphics[width=\linewidth]{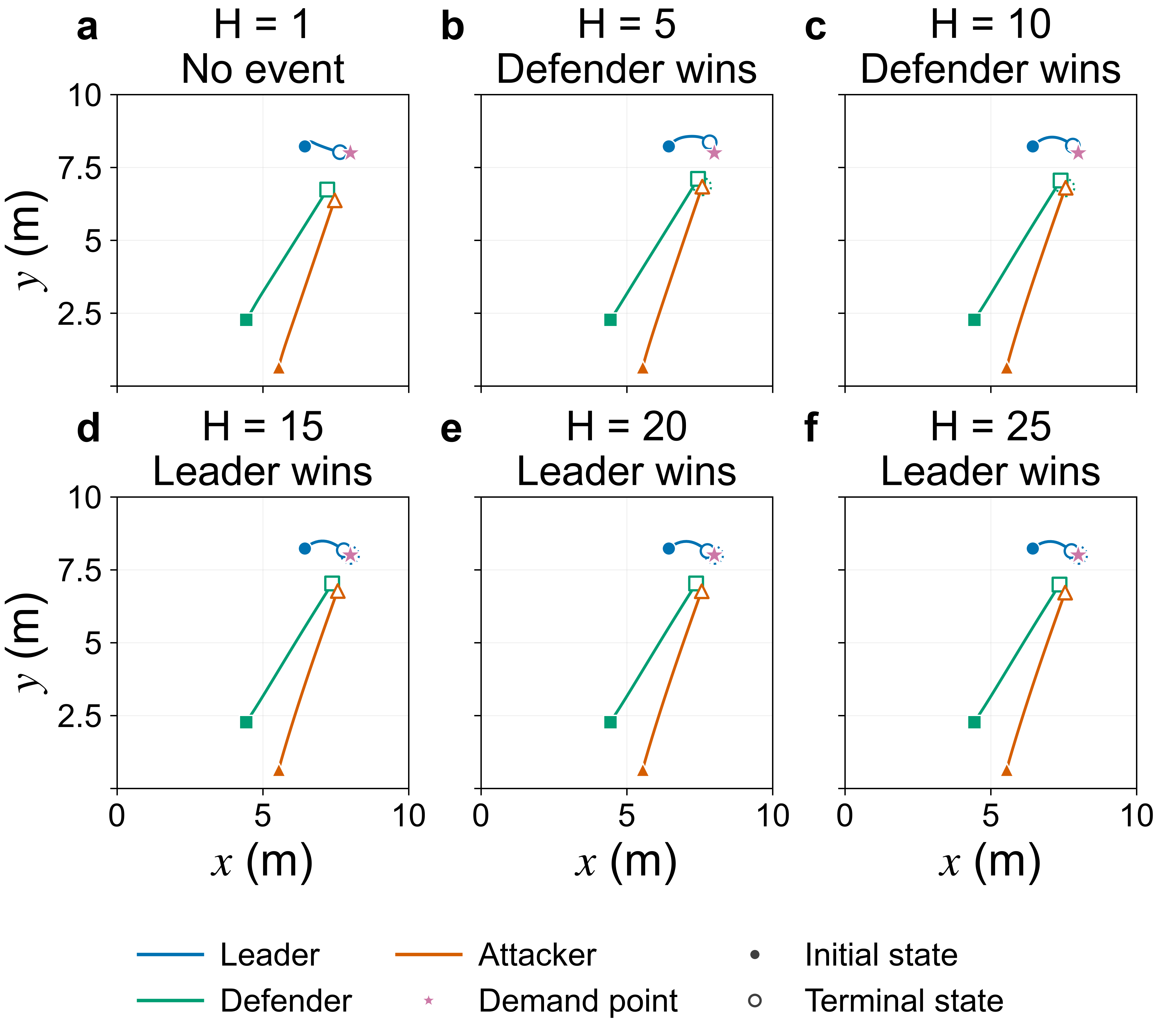}

\caption{Receding-horizon trajectories from the same initial configuration for $H=1,5,10,15,20,$ and $25$, respectively. No event occurs for $H=1$, the defender wins for $H=5$ and $10$, and the leader wins for $H=15$, $20$, and $25$.}
\label{fig:horizon_outcome_examples}
\end{figure}

\subsection{Comparison with the Independent-Control Baseline}

To assess the contribution of the proposed nested game
formulation, we compare it with an independent-control
baseline. Under this baseline, each agent derives its
finite-horizon control policy while treating the strategies
of the other agents as fixed. Consequently, the leader does
not incorporate the defender's best-response mapping into
its optimization.

Let $J^i_{\mathrm{sim}}$ denote the realized discounted
objective of agent $i$ under the independent-control
baseline, and let $J^i_{\mathrm{nested}}$ denote that under
the proposed nested formulation, where $i\in\{L,D,A\}$.
The objective difference is defined as
\begin{equation*}
    \Delta J^i = J^i_{\mathrm{sim}} - J^i_{\mathrm{nested}}.
\end{equation*}
Thus, $\Delta J^i>0$ indicates a lower cost under the nested formulation.

Figure~\ref{fig:independent_control_comparison} presents the
results for $H=5$. The leader difference is positive in all
50 trials, whereas the attacker difference is negative in
all trials. The defender difference varies across the
experiments. Therefore, incorporating the internal
leader--defender hierarchy consistently reduces the
leader's cost while increasing the attacker's cost,
indicating more effective coordination against the
attacker.

\begin{figure}[t]
\centering
\includegraphics[width=0.98\linewidth]{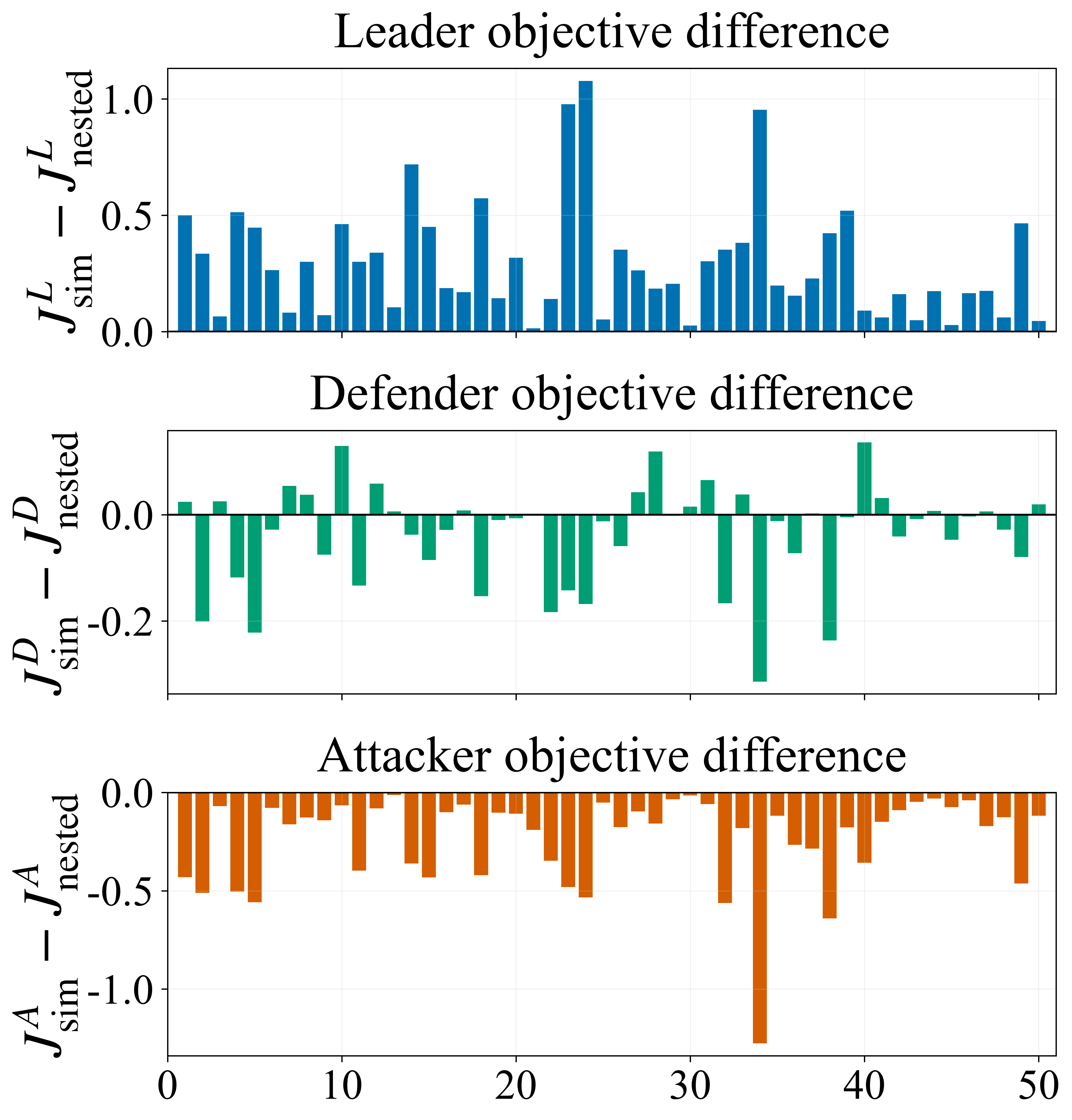}
\caption{Objective-function differences between the proposed
nested-equilibrium algorithm and the independent-control
baseline for $H=5$. Both methods use identical initial
conditions and model parameters.}
\label{fig:independent_control_comparison}
\end{figure}

\section{Conclusion}

This paper formulated a finite-horizon
leader--defender--attacker game and proposed a nested
equilibrium that combines an internal leader--follower
hierarchy with external noncooperative competition. For
linear--quadratic systems, we derived affine response laws,
coupled backward recursions, and sufficient conditions for
unique stagewise solutions. We further characterized the
policy-induced outcome regions through quadratic
inequalities in the initial joint state and implemented
the solver in a receding-horizon scheme. Numerical results
demonstrated the influence of the prediction horizon and
showed that the nested formulation reduced the leader's
realized cost while increasing the attacker's cost relative
to independent control. Future work will address nonlinear
dynamics, constraints, uncertainty, and theoretical
guarantees for receding-horizon execution.



\raggedbottom

\bibliographystyle{IEEEtran}
\bibliography{acc_introduction_references}

\end{document}